\documentclass[a4paper]{article}
\usepackage{amsmath,amssymb,array,url}
\usepackage{graphicx,booktabs}
\usepackage{multirow}
\usepackage{amsthm}
    \newtheorem{aaaaa}{DO NOT USE!}
    \newtheorem{proposition}[aaaaa]{Proposition}
    \newtheorem{corollary}[aaaaa]{Corollary}
    \newtheorem{remark}[aaaaa]{Remark}
    \newtheorem{definition}[aaaaa]{Definition}
    \newtheorem{example}[aaaaa]{Example}
    \newtheorem{theorem}[aaaaa]{Theorem}
    \newtheorem{algorithm}[aaaaa]{Algorithm}

\begin{document}

\title{List decoding of a class of affine variety codes}

\author{Olav Geil and Casper Thomsen\\
Department of Mathematical Sciences\\
Aalborg University, Denmark\\
\url{olav@math.aau.dk}, \url{caspert@math.aau.dk}
}

\maketitle  

\begin{abstract}
 Consider a polynomial $F$ in $m$ variables and a finite point
 ensemble $S=S_1\times \cdots \times S_m$. When given the leading monomial of $F$ with
 respect to a lexicographic ordering we derive improved information on
 the possible number of zeros of $F$ of multiplicity at least $r$ from
 $S$. We then use this information to design a list decoding algorithm
 for a large class of affine variety codes.\\

\end{abstract}

\section{Introduction}\label{secintro}
In this paper we study affine variety codes over ${\mathbf{F}}_q$ in
the special case where the variety is $S_1 \times \cdots \times
S_m$. Here, $S_i \subseteq {\mathbf{F}}_q$, $i=1, \ldots , m$. More
formally  write
$$S=S_1 \times \cdots \times S_m=\{P_1, \ldots ,P_{|S|}\}$$ 
and consider the evaluation map 
$${\mbox{ev}}_S:{\mathbf{F}}_q[X_1,
\ldots ,X_m] \rightarrow {\mathbf{F}}_q^{|S|}, {\mbox{ \ \ \ \ }}{\mbox{ev}}_S(F)=(F(P_1),
\ldots , F(P_{|S|})).$$ 
Let 
$${\mathbb{M}}\subseteq \{X_1^{i_1} \cdots X_m^{i_m} \mid 0 \leq i_j <
|S_j|, j=1, \ldots , m\}$$
and define the affine variety code
$$E({\mathbb{M}},S)={\mbox{Span}}_{{\mathbf{F}}_q}\{{\mbox{ev}}_S(M)
\mid M \in {\mathbb{M}} \}.$$
Clearly, the generalized Reed-Muller codes is a special example of
such a code. Other well-known examples are the hyperbolic codes \cite{hyperbolic},
Reed-Solomon product codes \cite{augot}, (generalized)
toric codes \cite{hansen,ruano,brasosul}, and Joyner codes
\cite{joyner}. In the present work we show that the dimension of
$E({\mathbb{M}},S)$ equals $|{\mathbb{M}}|$ and we present a lower
bound on the 
minimum distance which turns out to be sharp when $|{\mathbb{M}}|$
satisfies certain reasonable criteria. We then present a
 decoding algorithm for 
$E({\mathbb{M}},S)$. The algorithm is a straightforward
generalization of the Guruswami-Sudan decoding algorithm for
Reed-Solomon codes \cite{GS}, except that it involves a preparation step. A main
ingredient in the preparation step is information about how many
zeros of multiplicity at least $r$ a polynomial $F(X_1, \ldots
,X_m)$ can have given information about its leading monomial with respect
to a lexicographic ordering. \\
It is well-known that the generalized Reed-Muller codes can be viewed as subfield subcodes
of Reed-Solomon codes. As demonstrated by Pellikaan and Wu \cite{pw_ieee,PellikaanWu} this
observation leads
to an efficient decoding algorithm of generalized Reed-Muller codes
via the Guruswami-Sudan decoding algorithm for Reed-Solomon
codes. When the variety does not equal $S={\mathbf{F}}_q \times \cdots
\times {\mathbf{F}}_q $, or when ${\mathbb{M}}$ is not chosen as all the monomials of
degree up to some number, such a neat result does not hold
any more.\\
Pellikaan and Wu also presented a direct interpretation of the
Guruswami-Sudan decoding algorithm for generalized Reed-Muller
codes. Even though Augot et al. improved the analysis of the
latter algorithm dramatically \cite{augot,augotstepanov} the algorithm that uses the subfield
subcode is still superior when decoding generalized Reed-Muller
codes. The analysis by Augot et al.\ uses a 
generalization of the Schwartz-Zippel bound to also deal with
multiplicity. A proof of this bound was given  by Dvir et al.\
in~\cite{dvir}. Augot et al.\  then used their insight to
generalize the second algorithm by Pellikaan and Wu (the direct
interpretation of the Guruwami-Sudan algorithm) to Reed-Muller
like codes over $S=S_1 \times \cdots \times S_m$, when $S_1=\cdots
=S_m$, and to Reed-Solomon product codes as well. Based on the
generalized 
Schwartz-Zippel bound they estimated the decoding radius. \\
In the
present work we improve upon the methods from~\cite{dvir} to derive
more detailed information on how many zeros of prescribed multiplicity
 a polynomial can have. Rather than using only information on the
degree of the polynomial we use information on the leading monomial
with respect to a lexicographic ordering. We present both a recursive
algorithm to find such bounds and also closed formula
expressions for the case of polynomials in two variables. \\
The
interpolation polynomial $Q(X_1, \ldots , X_m,Z)$ in the list decoding
algorithm of the present paper can be viewed as a polynomial in $Z$ with coefficients from
${\mathbf{F}}_q[X_1, \ldots , X_m]$. Having fixed the code to be used,
in a preparation step we determine sets from which the coefficients
must be taken. Here, we use our improved insight on how many zeros
with prescribed multiplicity a polynomial can have given information
about 
its leading monomial. The idea of a preparation step comes
from~\cite{geilmatsumoto} where an interpretation of the
Guruswami-Sudan decoding algorithm without multiplicity was described
for order domain codes. Our experiments show that the method of the
present work is an improvement upon the situation where only the generalized Schwartz-Zippel bound
is used for the design and the analysis of the decoding algorithm. Our
algorithm works for codes of not too high dimensions. For small
dimensions the algorithm often  decodes more than $d/2$
errors. 
{\section{Parameters of the codes}\label{secto}}
Throughout the paper we shall use the notation
$s_i=|S_i|$ for $i=1, \ldots , m$. Also we shall write $n=|S|$ as
clearly the length of $E({\mathbb{M}},S)$ equals $|S|$. First we show
how to find the dimension of the code. 
\begin{proposition}
The dimension of $E({\mathbb{M}},S)$ equals $|{\mathbb{M}}|$.
\end{proposition}
\begin{proof}
We only need to show that 
$$\{ {\mbox{ev}}_S(X_1^{i_1}, \ldots , X_m^{i_m}) \mid 
 0 \leq i_j <
s_j, j=1, \ldots , m\}$$
constitutes a basis for ${\mathbf{F}}_q^{n}$ as a vectorspace over
${\mathbf{F}}_q$. For this purpose it is sufficient to show that the 
restriction of ${\mbox{ev}}$ to 
\begin{eqnarray}
\{ G(X_1, \ldots , X_m) \mid \deg_{X_i}(G) < s_i, i=1, \ldots ,m\} \label{equi}
\end{eqnarray}
is surjective. Given $(a_1, \ldots ,a_{n})\in {\mathbf{F}}_q^{n}$ let
$$F(X_1, \ldots , X_m)=\sum_{v=1}^n a_v \prod_{i=1}^m \prod_{a \in {\mathbf{F}}_q
  \backslash \{P_i^{(v)}\}}\bigg( \frac{X_i-a}{P_i^{(v)}-a} \bigg).$$
Here, we have used the notation $P_v=(P_1^{(v)}, \ldots , P_n^{(v)})$,
$v=1, \ldots , n$. 
It is clear that ${\mbox{ev}}_S(F)=(a_1, \ldots , a_n)$ and therefore
${\mbox{ev}}_S : {\mathbf{F}}_q [X_1, \ldots , X_m] \rightarrow
{\mathbf{F}}_q^n$ is surjective. Consider an arbitrary monomial
ordering. Let $R(X_1, \ldots , X_m)$ be the remainder of $F(X_1,
\ldots ,X_m)$ after division with
$$\{ A_1(X_1, \ldots , X_m), \ldots , A_m(X_1, \ldots , X_m)\}.$$
Here, $A_i(X_1, \ldots ,X_m)=\prod_{a \in S_i}(X_i-a)$, $i=1, \ldots
,m$. Clearly, $F(P_i)=R(P_i)=a_i$, $i=1, \ldots , n$. Hence, the
restriction of ${\mbox{ev}}_S$ to~(\ref{equi}) is indeed surjective.
\end{proof}
We next show how to estimate the minimum distance of
$E({\mathbb{M}},S)$. The Schwartz-Zippel bound \cite{Schwartz,Zippel,DeMilloLipton} is as follows:
\begin{theorem}\label{thorgsz}
Given a lexicographic ordering let the leading monomial of $F(X_1,
\ldots , X_m)$ be $X_1^{i_1}\cdots X_m^{i_m}$.
The number of elements in
$S=S_1 \times \cdots \times S_m$ that are zeros of $F$ is at most
equal to 
$$i_1s_2\cdots s_m+s_1 i_2 s_3 \cdots s_m+\cdots
+s_1\cdots s_{m-1} i_m.$$
\end{theorem}
The proof of this result is purely combinatorial. Using the
inclusion-exclusion principle it can actually be strengthened to the
following result which is a special case of the footprint bound from Gr\"{o}bner basis
theory: 
\begin{theorem}\label{footprspecial}
Given a lexicographic ordering let the leading monomial of $F(X_1,
\ldots , X_m)$ be $X_1^{i_1}\cdots X_m^{i_m}$. The number of elements in
$S=S_1 \times \cdots \times S_m$ that are zeros of $F$ is at most
equal to $$ n-(s_1-i_1)(s_2-i_2)\cdots (s_m-i_m). $$
\end{theorem}
\begin{proposition}
The minimum distance of $E({\mathbb{M}},S)$ is at least 
$$\min \{ (s_1-i_1)(s_2-i_2)\cdots (s_m-i_m) | X_1^{i_1}\cdots
X_m^{i_m} \in {\mathbb{M}}\}.$$
If for every $M\in {\mathbb{M}}$ it holds that $N \in {\mathbb{M}}$ for
all $N$ such that $N | M$ then the bound is sharp.
\end{proposition}
\begin{proof}
The first part follows from Theorem~\ref{footprspecial}. To see the
last part write for $i=1, \ldots , m$, $S_i=\{b_1^{(i)}, \ldots , b_{|S_i|}^{(i)}\}$.
The polynomial 
$$F(X_1, \ldots  X_m)=\prod_{v=1}^m \prod_{j=1}^{i_v} \big(X_v-b_j^{(v)}\big)$$
has leading monomial $X_1^{i_1}\cdots X_m^{i_m}$ with respect to any
monomial ordering and evaluates to zero in exactly $n- (s_1-i_1)(s_2-i_2)\cdots (s_m-i_m)$ points from $S$. Finally, any monomial that occurs in the support
of $F$ is a factor of $X_1^{i_1} \cdots X_m^{i_m}$.
\end{proof}
\section{Bounding the number of zeros of multiplicity~$r$}
In the following let $\vec{X}=(X_1, \ldots , X_m)$ and $\vec{T}=(T_1,
\ldots , T_m)$.
\begin{definition}
Let ${\mathbf{F}}$ be any field. Given $F(\vec{X})\in {\mathbf{F}}[\vec{X}]$ and
$\vec{k} \in {\mathbf{N}}_0^m$ then the $\vec{k}$'th
Hasse derivative of $F$, denoted by $F^{(\vec{k})}(\vec{X})$, is the
coefficient of $\vec{T}^{\vec{k}}$ in $F(\vec{X}+\vec{T})$. In other words 
$$F(\vec{X}+\vec{T})=\sum_{\vec{k}} F^{(\vec{k})}(\vec{X})\vec{T}^{\vec{k}}.$$
\end{definition}
The concept of multiplicity for univariate polynomials is generalized
to multivariate polynomials in the following way:
\begin{definition}\label{defmult}
For $F(\vec{X}) \in {\mathbf{F}}[\vec{X}]\backslash \{ {0} \}$ and
$\vec{a}\in {\mathbf{F}}^m$ we define the multiplicity of $F$ at $\vec{a}$
denoted by ${\mbox{mult}}(F,\vec{a})$ as follows: Let $r$ be an
integer such that for every $\vec{k}=(k_1, \ldots ,
k_m) \in {\mathbf{N}}_0^m$ with $k_1+\cdots +k_m < r$, $F^{(\vec{k})}(\vec{a})=0$
  holds, but for some  $\vec{k}=(k_1, \ldots ,
k_m) \in {\mathbf{N}}_0^m$ with $k_1+\cdots +k_m = r$,
$F^{(\vec{k})}(\vec{a})\neq 0$ holds, then 
${\mbox{mult}}(F,\vec{a})=r$. If $F=0$ then we define ${\mbox{mult}}(F,\vec{a})=\infty$.
\end{definition}
Elaborating on the results in~\cite{dvir} we find:
\begin{theorem}\label{prop-sz-gen}
Let $F(\vec{X}) \in {\mathbf{F}}[\vec{X}]$ be a non-zero polynomial and
let $X_1^{i_1} \cdots X_m^{i_m}$ be its leading
monomial with respect to a lexicographic ordering. Then for
any finite sets $S_1, \ldots ,S_m \subseteq {\mathbf{F}}$
\begin{eqnarray}
\sum_{\vec{a}\in S_1 \times \cdots \times S_m}{\mbox{mult}}(F,\vec{a})
\leq i_1s_2\cdots s_m+s_1i_2s_3 \cdots s_m+\cdots +s_1\cdots s_{m-1}i_m.\nonumber
\end{eqnarray}
\end{theorem}
Theorem~\ref{thorgsz} now generalizes to the
following result which we call the Schwartz-Zippel bound.
\begin{corollary}\label{cor-sz-gen}
Let $F(\vec{X}) \in {\mathbf{F}}[\vec{X}]$ be a non-zero polynomial and
let $X_1^{i_1} \cdots X_m^{i_m}$ be its leading
monomial with respect to a lexicographic ordering. Assume $S_1, \ldots
,S_m \subseteq {\mathbf{F}}$ are finite sets. 
Then over
$S_1 \times \cdots \times S_m $ the number
of zeros of multiplicity at least $r$ is less than or equal to 
\begin{eqnarray}
\big( i_1s_2\cdots s_m+s_1i_2s_3\cdots s_m+\cdots +s_1\cdots
s_{m-1}i_m \big) /r.\label{szbound}
\end{eqnarray}
\end{corollary}
Just as we were able to improve upon Theorem~\ref{thorgsz} by using the
inclusion-exclusion principle we will also be able to improve upon
Corollary~\ref{cor-sz-gen}. However, now the situation is much more
complex and therefore the inclusion-exclusion principle is no longer
sufficient. What we need to strengthen Corollary~\ref{cor-sz-gen} is
the following rather technical function: 
\begin{definition}\label{defD}
Let $r \in {\mathbf{N}}, i_1, \ldots , i_m \in {\mathbf{N}}_0$. Define 
$$D(i_1,r,s_1)=\min \big\{\big\lfloor \frac{i_1}{r} \big\rfloor,s_1\big\}$$
and for $m \geq 2$
\begin{multline*}
D(i_1, \ldots , i_m,r,s_1, \ldots ,s_m)=
\\
\begin{split}
\max_{(u_1, \ldots  ,u_r)\in A(i_m,r,s_m) }&\bigg\{ (s_m-u_1-\cdots -u_r)D(i_1,\ldots ,i_{m-1},r,s_1,
\ldots ,s_{m-1})\\
&\quad+u_1D(i_1, \ldots , i_{m-1},r-1,s_1, \ldots ,s_{m-1})+\cdots
\\
&\quad +u_{r-1}D(i_1, \ldots ,i_{m-1},1,s_1, \ldots , s_{m-1})+u_rs_1\cdots
s_{m-1} \bigg\}
\end{split}
\end{multline*}
where 
\begin{multline}
A(i_m,r,s_m)= \nonumber \\
\{ (u_1, \ldots , u_r) \in {\mathbf{N}}_0^r \mid u_1+ \cdots
+u_r \leq s_m {\mbox{ \ and \ }} u_1+2u_2+\cdots +ru_r \leq i_m\}.\nonumber
\end{multline}
\end{definition}
\begin{theorem}\label{prorec}
For a polynomial $F(\vec{X})\in {\mathbf{F}}[\vec{X}]$ let $X_1^{i_1}\cdots X_m^{i_m}$ be its leading monomial with
respect to the lexicographic ordering with $X_m\prec \cdots \prec X_1$. Then $F$ has at most $D(i_1, \ldots , i_m,r,s_1,
\ldots ,s_m)$ zeros of multiplicity at least $r$ in $S_1\times \cdots
\times S_m$. The corresponding recursive algorithm produces a number
that is at most equal to the number found in
Corollary~\ref{cor-sz-gen} and at most equal to $s_1 \cdots s_m$.
\end{theorem}
\begin{proof}
The proof involves a modification of the method described in~\cite{dvir}. Due to lack of space we do not include it here. 
\end{proof}
\begin{remark}\label{rembig}
Given $(i_1, \ldots , i_m, r, s_1, \ldots ,s_m)$ with $\lfloor
i_1/s_1\rfloor+\cdots +\lfloor i_m/s_m\rfloor \geq r$ then there exist
polynomials with the leading monomial being $X_1^{i_1} \cdots X_m^{i_m}$ such that all points in $S_1 \times \cdots \times S_m$ are
zeros of multiplicity at least $r$. 
Hence, we need only apply
the algorithm to tuples $(i_1, \ldots , i_m)$ such that
\begin{equation}
 \lfloor
i_1/s_1\rfloor+\cdots +\lfloor i_m/s_m\rfloor < r.       \label{eqhak}
\end{equation}
\end{remark}
In a
number of experiments listed in~\cite{hmpage} we calculated the value $D(i_1, \ldots ,
i_m,r,q, \ldots ,q)$ for various choices of $m$, $q$ and $r$ and for all values of
$(i_1, \ldots , i_m)$ satisfying~(\ref{eqhak}).
We here list the
maximal attained improvement obtained by using Proposition~\ref{prorec} rather
than using Corollary~\ref{cor-sz-gen}. We do this relatively to the number of points $q^m$. In other words we list in Table~\ref{tabny5} the value
$$\bigg(\max_{i_1, \ldots , i_m} \{\min\{ (i_1+\cdots
i_m)q^{m-1}/r,q^m\}-D(i_1, \ldots , i_m,r,q, \ldots , q)\}\bigg)/q^m$$
for various choices of $m, q, r$. \\
The experiments also show  distinct average improvement. This is
illustrated in Table~\ref{tabendnuenny1} where for fixed $q, r, m$ we list
the mean value of
\begin{equation}
\frac{ \min \{(i_1+\cdots + i_m)q^{m-1},q^m\}-D(i_1, \ldots ,
i_m,r,q, \ldots , q)}{\min \{(i_1+\cdots +
i_m)q^{m-1},q^m\}}.\label{eqangle}
\end{equation}
The average is taken over the set of exponents $(i_1, \ldots ,
i_m)\neq \vec{0}$ where 
$\lfloor i_1/q\rfloor+\cdots +\lfloor i_m/q\rfloor
< r$  
holds.

\begin{table}
\centering
\caption{Maximum improvements relative to $q^m$; truncated}
\newcommand{\SP}{~~}
\begin{tabular}{@{}c@{}c@{~~~~}r@{.}l@{\SP}r@{.}l@{\SP}r@{.}l@{\SP}r@{.}l@{\SP}r@{.}l@{\SP}r@{.}l@{\SP}r@{.}l@{\SP}r@{.}l@{\SP}r@{.}l@{\SP}r@{.}l@{}}
\toprule
$m$&& \multicolumn{8}{c}{2} & \multicolumn{8}{c}{3} & \multicolumn{4}{c}{4} \\
\cmidrule(r){3-10} \cmidrule(r){11-18} \cmidrule{19-22}
$r$&&\multicolumn{2}{l}{2}&\multicolumn{2}{l}{3}&\multicolumn{2}{l}{4}&\multicolumn{2}{l}{5}&\multicolumn{2}{l}{2}&\multicolumn{2}{l}{3}&\multicolumn{2}{l}{4}&\multicolumn{2}{l}{5}&\multicolumn{2}{l}{2}&\multicolumn{2}{l}{3}\\
\addlinespace
\multirow{7}{*}{$q$}
&\multicolumn{1}{c}{2}&0&25 &0&25 &0&25 &0&25 &0&25 &0&375&0&375&0&375&0&312&0&375\\
&\multicolumn{1}{c}{3}&0&222&0&222&0&222&0&222&0&296&0&296&0&296&0&296&0&296&0&333\\
&\multicolumn{1}{c}{4}&0&187&0&187&0&187&0&187&0&281&0&25 &0&25 &0&265&0&316&0&289\\
&\multicolumn{1}{c}{5}&0&24 &0&16 &0&16 &0&2  &0&256&0&256&0&232&0&24 &0&307&0&288\\
&\multicolumn{1}{c}{7}&0&204&0&204&0&163&0&142&0&279&0&244&0&227&0&209&0&299&0&276\\
&\multicolumn{1}{c}{8}&0&234&0&203&0&171&0&140&0&275&0&25 &0&214&0&203&0&299&0&275\\  
\bottomrule
\end{tabular}
\label{tabny5}
\end{table}
\begin{table}[!h]
\centering
\caption{The mean value of (\ref{eqangle}); truncated}
\newcommand{\SP}{~~}
\begin{tabular}{@{}c@{}c@{~~~~}r@{.}l@{\SP}r@{.}l@{\SP}r@{.}l@{\SP}r@{.}l@{\SP}r@{.}l@{\SP}r@{.}l@{\SP}r@{.}l@{\SP}r@{.}l@{\SP}r@{.}l@{\SP}r@{.}l@{}}
\toprule
$m$&& \multicolumn{8}{c}{2} & \multicolumn{8}{c}{3} & \multicolumn{4}{c}{4} \\
\cmidrule(r){3-10} \cmidrule(r){11-18} \cmidrule{19-22}
$r$&&\multicolumn{2}{l}{2}&\multicolumn{2}{l}{3}&\multicolumn{2}{l}{4}&\multicolumn{2}{l}{5}&\multicolumn{2}{l}{2}&\multicolumn{2}{l}{3}&\multicolumn{2}{l}{4}&\multicolumn{2}{l}{5}&\multicolumn{2}{l}{2}&\multicolumn{2}{l}{3}\\
\addlinespace
\multirow{7}{*}{$q$}
&\multicolumn{1}{c}{2}&0&363&0&273&0&337&0&291&0&301&0&300&0&342&0&307&0&248&0&260\\
&\multicolumn{1}{c}{3}&0&217&0&286&0&228&0&236&0&194&0&224&0&213&0&214&0&158&0&177\\
&\multicolumn{1}{c}{4}&0&191&0&197&0&232&0&195&0&158&0&169&0&180&0&172&0&125&0&135\\
&\multicolumn{1}{c}{5}&0&155&0&167&0&174&0&197&0&139&0&145&0&148&0&153&0&110&0&116\\
&\multicolumn{1}{c}{7}&0&128&0&137&0&138&0&138&0&119&0&122&0&121&0&119&0&093&0&098\\
&\multicolumn{1}{c}{8}&0&126&0&127&0&134&0&126&0&114&0&115&0&113&0&111&0&089&0&093\\
\bottomrule
\end{tabular}
\label{tabendnuenny1}
\end{table}
The values $D(i_1, \ldots ,i_m,r, s_1, \ldots , s_m)$ may sometimes be
time consuming to
calculate. Therefore it is relevant to have some closed formula
estimates of these numbers. We next present such estimates for the
case of two variables. By
Remark~\ref{rembig} the following proposition covers all non-trivial cases:\\
\begin{proposition}\label{protwovar}
For $k=1, \ldots , r-1$,  $D(i_1,i_2,r,s_1,s_2)$ is upper bounded by\\
$\begin{array}{cl}
{\mbox{(C.1)}}&  {\displaystyle{s_2\frac{i_1}{r}+\frac{i_2}{r}\frac{i_1}{r-k}}}\\
&{\mbox{if \  }}(r-k)\frac{r}{r+1}s_1 \leq i_1 < (r-k)s_1
{\mbox{ \ and \ }} 0\leq i_2 <ks_2,\\
{\mbox{(C.2)}}&
  {\displaystyle{s_2\frac{i_1}{r}+((k+1)s_2-i_2)(\frac{i_1}{r-k}-\frac{i_1}{r})+(i_2-ks_2)(s_1-\frac{i_1}{r})}}\\
& {\mbox{if \ }}(r-k)\frac{r}{r+1}s_1 \leq i_1 < (r-k)s_1 {\mbox{ \
    and \ }} ks_2\leq i_2 <(k+1)s_2,\\
{\mbox{(C.3)}}&
{\displaystyle{s_2\frac{i_1}{r}+\frac{i_2}{k+1}(s_1-\frac{i_1}{r})}}\\
&{\mbox{if \ }} (r-k-1)s_1 \leq i_1 < (r-k)\frac{r}{r+1}s_1 {\mbox{ \
    and \ }} 0 \leq i_2 < (k+1)s_2.
\end{array}
$\\
Finally,\\
$\begin{array}{cl}
{\mbox{(C.4)}}& {\displaystyle{D(i_1,i_2,r,s_1,s_2)=s_2\lfloor \frac{i_1}{r} \rfloor
  +i_2(s_1-\lfloor \frac{i_1}{r} \rfloor )}}\\
& {\mbox{if \ }} s_1(r-1) \leq i_1 < s_1r {\mbox{ \ and \ }} 0 \leq i_2 < s_2.
\end{array}
$\\
The above numbers are at most equal to $\min\{(i_1s_2+s_1i_2)/r, s_1s_2 \}$.
\end{proposition}
\begin{proof}
The estimates are found by treating $i_1$ and $i_2$ as rational
numbers rather than integers. Due to lack of space we do not give the
details here.
\end{proof}

{\section{The decoding algorithm}\label{secfire}}
\noindent
The main ingredient of the decoding algorithm is to find an
interpolation polynomial 
$$Q(X_1, \ldots , X_m,Z)=Q_0(X_1, \ldots , X_m)+Q_1(X_1,\ldots ,
X_m)Z+\cdots +Q_t(X_1, \ldots ,X_m)Z^t$$
such that $Q(X_1, \ldots , X_m,F(X_1, \ldots  ,X_m))$ cannot have more
than $n-E$ different zeros of multiplicity at least $r$ whenever
${\mbox{Supp}}(F)\subseteq {\mathbb{M}}$. The integer $E$ above is the
number of errors to be corrected by our list decoding algorithm. To
fulfill this requirement we will define appropriate sets of monomials
$B(i,E,r)$, $i=1, \ldots , t$ and then require $Q_i(X_1, \ldots
,X_m)$ to be chosen such that ${\mbox{Supp}}(Q_i)\subseteq
B(i,E,r)$. Rather than using the results from the previous section on
all possible choices of $F(X_1, \ldots  ,X_m)$ with ${\mbox{Supp}}(F)
\subseteq {\mathbb{M}}$ we only consider the worst cases where the
leading monomial of $F$ is contained in the following set:
\begin{definition}
$$\overline{{\mathbb{M}}}=\{ M \in {\mathbb{M}} \mid {\mbox{ \ if \ }} N\in
{\mathbb{M}} {\mbox{ \ and \ }} M|N {\mbox{ \ then \ }} M=N\}.$$
\end{definition}
Hence, $\overline{{\mathbb{M}}}$ is so to speak the boarder of
${\mathbb{M}}$. 
\begin{definition}
Given positive integers $i,E,r$ with $E<n$ let
$$B(i,E,r)=\{K \in \Delta(r,m) \mid D_r(K M^{i}) < n-E {\mbox{ \ for
    all \ }} M \in \overline{{\mathbb{M}}} \}.$$
Here, $D_r(X_1^{i_1}, \ldots , X_m^{i_m})$ can either be $D(i_1, \ldots
, i_m,r,s_1, \ldots ,s_m)$ or in the case of two variables it can be
the numbers from Proposition~\ref{protwovar}. Another option would be
to let $D_r(X_1^{i_1}, \ldots , X_m^{i_m})$ be the number in~(\ref{szbound}). With a reference to
Remark~\ref{rembig} we have defined
$$\Delta(r,m)=\{ X_1^{i_1}\cdots X_m^{i_m} \mid  \lfloor
i_1/s_1\rfloor+\cdots +\lfloor i_m/s_m\rfloor < r\}.
$$
\end{definition}
The decoding algorithm calls for positive integers $t,E,r$ such that
\begin{equation}
\sum_{i=1}^t |B(i,E,r)| >n N(m,r), \label{eqsnabel}
\end{equation}
where $N(m,r)={{m+r}\choose{m+1}}$ is the number of linear equations to be satisfied
for a point in ${\mathbf{F}}_q^{m+1}$ to be a zero of $Q(X_1, \ldots ,
X_m,Z)$ of multiplicity at least $r$. 
As we will see condition~(\ref{eqsnabel}) ensures that we
can correct $E$ errors. We will say that $(t,E,r)$ satisfies the
initial condition if given the pair $(E,r)$, $t$ is the smallest integer such
that~(\ref{eqsnabel}) is satisfied. Whenever this is the case we
define $B^\prime(t,E,r)$ to be any subset of $B(t,E,r)$ such that 
$$\sum_{i=1}^{t-1}|B(i,E,r)| + |B^\prime(t,E,r)|=n N(m,r)+1.$$
Replacing $B(t,E,r)$ with $B^\prime(t,E,r)$ will lower the run time
of the algorithm.

\begin{algorithm}\label{thealgorithm}
\noindent {\it{Input:}}\\
Received word $\vec{r}=(r_1, \ldots , r_n) \in {\mathbf{F}}_q^n$.\\
Set of integers $(t,E,r)$ that satisfies the initial condition.\\
Corresponding sets $B(1,E,r)\, \ldots ,B(t-1,E,r),B^\prime(t,E,r)$.\\

\noindent {\it{Step 1}}\\
Find non-zero polynomial
$$Q(X_1, \ldots , X_m Z)=Q_0(X_1, \ldots , X_m)+Q_1(X_1, \ldots
,X_m)Z+\cdots +Q_t(X_1, \ldots ,X_m)Z^t$$
such that
\begin{enumerate}
\item ${\mbox{Supp}}(Q_i) \subseteq B(i,E,r)$ for $i=1, \ldots ,t-1$
  and ${\mbox{Supp}}(Q_t) \subseteq B^\prime(t,E,r)$,
\item $(P_i,r_i)$ is a zero of $Q(X_1, \ldots , X_m,Z)$ of
  multiplicity at least $r$ for $i=1, \ldots ,n$.
\end{enumerate}

\noindent {\it{Step 2}}\\
Find all $F(X_1, \ldots ,X_m) \in {\mathbf{F}}_q[X_1, \ldots , X_m]$
such that 
\begin{equation}
(Z-F(X_1, \ldots ,X_m)) |Q(X_1, \ldots , X_m,Z).\label{eqstar}
\end{equation}

\noindent {\it{Output:}}\\
A list containing $(F(P_1), \ldots , F(P_n))$ for all $F$ satisfying~(\ref{eqstar}).
\end{algorithm}

\begin{theorem}
The output of Algorithm~\ref{thealgorithm} contains all words in
$E({\mathbb{M}},S)$ within distance $E$ from the received word
$\vec{r}$. Once the preparation step has been performed the algorithm runs in time ${\mathcal{O}}(\bar{n}^3)$ where
$\bar{n}=n{{m+r}\choose{m+1}}$. For given multiplicity $r$ the maximal
number of correctable errors $E$ and the corresponding sets $B(1,
E,r), \ldots ,B(t-1,E,r)$, $B^\prime(t,E,r)$ can be found in time
${\mathcal{O}}(n \log(n) r^m s' |\overline{\mathbb{M}}|/\sigma)$ 
assuming that the values of the function $D_r$ are known. Here
$\sigma = \max\{i_1+\dotsb+i_m|X_1^{i_1} \dotsb X_m^{i_m} \in
\overline{\mathbb{M}}\}$ and $s' = \max\{s_1, \dotsc, s_m\}$.
\end{theorem}
\begin{proof} The interpolation problem corresponds to $\bar{n}$
  homogeneous linear equations in $\bar{n}+1$ unknowns. 
Hence, indeed a suitable $Q$ can be found in time
${\mathcal{O}}(\bar{n}^3)$. Now assume ${\mbox{Supp}}(F) \subseteq {\mathbb{M}}$  and that
${\mbox{dist}}_H({\mbox{ev}}_S(F),\vec{r})\leq E$. Then $P_j$ is a zero of $Q(X_1, \ldots , X_m,F(X_1,
\ldots , X_m))$ of multiplicity at least $r$ for at least $n-E$ choices of $j$. By the
definition of $B(i,E,r)$ this can, however, only be the case if $Q(X_1, \ldots ,
X_m,F(X_1, \ldots , X_m))=0$. Therefore, $Z-F(X_1, \ldots ,X_m)$ is a
factor in $Q(X_1, \ldots ,X_m,Z)$. Finding linear factors of
polynomials in $({\mathbf{F}}_q[X_1, \ldots , X_m])[Z]$ can be done in
time ${\mathcal{O}}(\bar{n}^3)$ by applying Wu's algorithm in~\cite{wu} (see \cite[p.\ 21]{PellikaanWu}).\end{proof}
\section{Examples}\label{secexample}
\begin{example}\label{ex2}
In this example we consider a point ensemble $S=S_1 \times S_2$ with
$s_1=128$ and $s_2=64$. We consider codes $E({\mathbb{M}},S)$ where
${\mathbb{M}}=\{ X_1^{i_1}X_2^{i_2} \mid i_1+2i_2 \leq u\}$. Note, the
weight $2$ which corresponds to the number $s_1/s_2$. The performance of
Algorithm~\ref{thealgorithm} is independent of the field in which $S_1$
and $S_2$ live. In Table~\ref{tab:RMweighted128x64} we list the number of errors that we
can correct when the function $D_r(i_1,i_2)$ -- is chosen to be
$D(i_1,i_2,r,80,80)$ (column $D$), -- is the closed formula expression from
Proposition~\ref{protwovar} (column C), -- or is the Schwartz-Zippel bound (column S),
respectively. The row $\lfloor (d-1)/2\rfloor$
corresponds to half the minimum distance and the row Dim.\ is the dimension of
the code.
    \begin{table}
    \caption{Error correction capabilities for the codes in Example~\ref{ex2}}
    \label{tab:RMweighted128x64}
    \centering
    \scalebox{.90}{
    \newcommand{\SP}{~\hspace*{0.10em}}
    \begin{tabular}{r c@{\SP}c@{\SP}c c@{\SP}c@{\SP}c c@{\SP}c@{\SP}c c@{\SP}c@{\SP}c}
        \toprule
        $u$
        &\multicolumn{3}{c}{3}
        &\multicolumn{3}{c}{4}
        &\multicolumn{3}{c}{7}
        &\multicolumn{3}{c}{20}\\
        \cmidrule(lr){2-4}
        \cmidrule(lr){5-7}
        \cmidrule(lr){8-10}
        \cmidrule(lr){11-13}
        $r$&$D$&C&S &$D$&C&S &$D$&C&S &$D$&C&S \\
        \midrule
         2&5129&5105&4895&4799&4777&4575&4143&4124&3871&2487&2475&2175\\
         3&5367&5333&5205&5048&5016&4906&4407&4381&4245&2855&2833&2666\\
         4&5474&5438&5343&5180&5143&5071&4566&4535&4431&3060&3031&2927\\
         9&    &5653&5617&    &5390&5361&    &4817&4785&    &3415&3384\\
        20&    &5757&5740&    &5509&5494&    &4959&4943&    &3609&3599\\
        \midrule
        $\lfloor\frac{d-1}{2}\rfloor$
        &\multicolumn{3}{c}{3999}
        &\multicolumn{3}{c}{3967}
        &\multicolumn{3}{c}{3871}
        &\multicolumn{3}{c}{3455}\\
        Dim.
        &\multicolumn{3}{c}{6}
        &\multicolumn{3}{c}{9}
        &\multicolumn{3}{c}{20}
        &\multicolumn{3}{c}{121}\\
        \bottomrule
    \end{tabular}
    }
    \end{table}
\end{example}
\begin{example}\label{ex3}
This is a continuation of Example~\ref{ex2}. We consider the same
point ensemble but choose ${\mathbb{M}}=\{X_1^{i_1}X_2^{i_2} \mid i_1
< k_1, i_2 < k_2\}$. That is, we consider Reed-Solomon
product codes. In Table~\ref{tab:PCweighted128x64} we list the number of errors that can
be corrected by Algorithm~\ref{thealgorithm}.
    \begin{table}
    \caption{Error correction capabilities for the codes in Example~\ref{ex3}.}
    \label{tab:PCweighted128x64}
    \centering
    \scalebox{.90}{
    \newcommand{\SP}{~\hspace*{0.10em}}
    \begin{tabular}{r c@{\SP}c@{\SP}c c@{\SP}c@{\SP}c c@{\SP}c@{\SP}c c@{\SP}c@{\SP}c}
        \toprule
        $(k_1,k_2)$
        &\multicolumn{3}{c}{$(4,7)$}
        &\multicolumn{3}{c}{$(5,9)$}
        &\multicolumn{3}{c}{$(8,15)$}
        &\multicolumn{3}{c}{$(21,41)$}\\
        \cmidrule(lr){2-4}
        \cmidrule(lr){5-7}
        \cmidrule(lr){8-10}
        \cmidrule(lr){11-13}
        $r$&$D$&C&S &$D$&C&S &$D$&C&S &$D$&C&S \\
        \midrule
         2&4036&4015&3519&3655&3639&3071&2820&2808&2111&1061&1055&0   \\
         3&4289&4261&3903&3911&3885&3498&3077&3058&2602&1183&1171&533 \\
         4&4411&4381&4111&4042&4011&3727&3214&3187&2847&1310&1291&831 \\
         9&    &4598&4487&    &4244&4124&    &3455&3313&    &1567&1365\\
        20&    &4711&4662&    &4364&4310&    &3588&3526&    &1704&1615\\
        \midrule
        $\lfloor\frac{d-1}{2}\rfloor$
        &\multicolumn{3}{c}{3720}
        &\multicolumn{3}{c}{3599}
        &\multicolumn{3}{c}{3248}
        &\multicolumn{3}{c}{1935}\\
        Dim.
        &\multicolumn{3}{c}{28}
        &\multicolumn{3}{c}{45}
        &\multicolumn{3}{c}{120}
        &\multicolumn{3}{c}{861}\\
        \bottomrule
    \end{tabular}
    }
    \end{table}
\end{example}
\begin{example}\label{ex1}
In this example we consider a point ensemble $S=S_1 \times S_2$ with
$s_1=s_2=80$. We consider codes $E({\mathbb{M}},S)$ where
${\mathbb{M}}=\{X_1^{i_1}X_2^{i_2} \mid i_1+i_2 \leq u\}$ for various
values of $u$. That is, we consider Reed-Muller like codes. The performance of Algorithm~\ref{thealgorithm} is
independent on the field in which $S_1$ and $S_2$ live. In
Table~\ref{tab:RMdiagonal80x80} we list the number of errors that we
can correct when the function $D_r(i_1,i_2)$ -- is chosen to be
$D(i_1,i_2,r,80,80)$ (column $D$), -- is the closed formula expression from
Proposition~\ref{protwovar} (column C), -- or is the Schwartz-Zippel bound (column S),
respectively. Column A corresponds to a bound
from~\cite{augotstepanov} on what can be achieved by applying their 
algorithm. The row ${\mbox{A}}_\infty$ corresponds to what could
theoretically be achieved by the algorithm in~\cite{augotstepanov} if
one uses high enough multiplicity. \\
Assuming $S_1,S_2
\subseteq {\mathbf{F}}_{128}$, $E({\mathbb{M}},S)$ corresponds to a
puncturing of the generalized Reed-Muller code
${\mbox{RM}}_{128}(u,2)$. This suggest that as an alternative to using
Algorithm~\ref{thealgorithm} one could decode with respect to
${\mbox{RM}}_{128}(u,2)$ treating the punctured points as errors. The
best known algorithm to decode ${\mbox{RM}}_{128}(u,2)$ for the values of
$u$ considered in this example is the algorithm by
Pellikaan and Wu which uses the subfield subcode approach. Row 
PW of the table explains what can be achieved by this
alternative approach. 
    \begin{table}
    \caption{Error correction capabilities for the codes in Example~\ref{ex1}}
    \label{tab:RMdiagonal80x80}
    \centering
    \scalebox{.75}{
    \newcommand{\SP}{~\hspace*{0.10em}}
    \begin{tabular}{r c@{\SP}c@{\SP}c@{\SP}c c@{\SP}c@{\SP}c@{\SP}c c@{\SP}c@{\SP}c@{\SP}c c@{\SP}c@{\SP}c@{\SP}c}
        \toprule
        $u$
        &\multicolumn{4}{c}{3}
        &\multicolumn{4}{c}{4}
        &\multicolumn{4}{c}{7}
        &\multicolumn{4}{c}{20}\\
        \cmidrule(lr){2-5}
        \cmidrule(lr){6-9}
        \cmidrule(lr){10-13}
        \cmidrule(lr){14-17}
        $r$&$D$&C&S&A &$D$&C&S&A &$D$&C&S&A &$D$&C&S&A \\
        \midrule
         2&3594&3571&3399&3310&3317&3297&3119&2999&2693&2679&2479&2302&1279&1279& 999& 585\\
         3&3791&3765&3679&3604&3524&3499&3413&3323&2943&2918&2799&2692&1575&1559&1439&1138\\
         4&3899&3869&3799&3758&3647&3618&3559&3492&3080&3058&2979&2896&    &1728&1639&1428\\
         9&    &4072&4053&4027&    &3837&3813&3792&    &3315&3297&3253&    &2053&2035&1935\\
        20&    &4171&4163&4152&    &3946&3939&3926&    &3444&3435&3418&    &2219&2211&2169\\

        \midrule
        A$_\infty$
        &\multicolumn{4}{c}{4257}
        &\multicolumn{4}{c}{4042}
        &\multicolumn{4}{c}{3558}
        &\multicolumn{4}{c}{2368}\\
        PW
        &\multicolumn{4}{c}{3891}
        &\multicolumn{4}{c}{3503}
        &\multicolumn{4}{c}{2568}
        &\multicolumn{4}{c}{0}\\ 
        $\lfloor\frac{d-1}{2}\rfloor$
        &\multicolumn{4}{c}{3079}
        &\multicolumn{4}{c}{3039}
        &\multicolumn{4}{c}{2919}
        &\multicolumn{4}{c}{2399}\\
        Dim.
        &\multicolumn{4}{c}{10}
        &\multicolumn{4}{c}{15}
        &\multicolumn{4}{c}{36}
        &\multicolumn{4}{c}{231}\\
        \bottomrule
    \end{tabular}
    }
    \end{table}
\end{example}
\section{Acknowledgments}\label{}
This work was supported in part by Danish Natural Research Council grant 272-07-0266. The authors would like to
thank Peter Beelen and Teo Mora for pleasant discussions. Also thanks
to L.\ Grubbe Nielsen for linguistical assistance.

\end{document}